\documentclass{article}
\usepackage{graphicx,latexsym,natbib,amssymb,amsmath, complexity, amsthm}
\usepackage{subfig, tikz, float}
\usetikzlibrary{calc}
\newcommand{\ignore}[1]{}%

\newtheorem{thm}{Theorem}

\newtheorem{lemma}{Lemma}

\theoremstyle{definition}
\newtheorem{defn}{Definition}
\theoremstyle{definition}
\newtheorem{rem}{Remark}

\newenvironment{theorem_again}[1]{\noindent{\bf Theorem #1.~}}{}
\newenvironment{lemma_again}[1]{\noindent{\bf Lemma #1.~}}{}

\begin{document}
\clubpenalty=10000 
\widowpenalty    = 10000 
%\title{A Practical Automated Market Maker\\ with Adaptive Liquidity}
\title{The Complexity of Fairness
through Equilibrium
}

\author{ 	{Abraham Othman}\\
           	{aothman@cs.cmu.edu}
	\and
		{Christos Papadimitriou}\\
        	{EECS}\\
        	{UC Berkeley}\\
        	{christos@cs.berkeley.edu}\\
	\and
		{Aviad Rubinstein}\\
        	{EECS}\\
        	{UC Berkeley}\\
        	{aviad@eecs.berkeley.edu}\\
}

\maketitle

\begin{abstract}
Competitive equilibrium with equal incomes (CEEI) is a well-known fair allocation mechanism \citep{Foley67:Resource,Varian74:Equity,Thomson85:Theories}; however, for indivisible resources a CEEI may not exist.   It was shown in \citet{ACEEI_Bud11} that in the case of indivisible resources there is always an allocation, called A-CEEI,  that is approximately fair, approximately truthful, and approximately efficient, for some favorable approximation parameters.  This approximation is used in practice to assign business school students to classes. 
%
%It is known that an approximate competitive equilibrium always exists, regardless of input preferences. 
%
In this paper we show that finding the A-CEEI allocation guaranteed to exist by Budish's theorem is $\PPAD$-complete.  We further show that finding an approximate equilibrium with better approximation  guarantees is even harder: $\NP$-complete.
\end{abstract}

\thispagestyle{empty}
\clearpage
\setcounter{page}{1}

%\category{J.4}{Social and Behavioral Sciences}{Economics}

%A category including the fourth, optional field follows...
%\category{D.2.8}{Software Engineering}{Metrics}[complexity measures, performance measures]

%General terms should be selected from the following 16 terms: Algorithms, Management, Measurement, Documentation, Performance, Design, Economics, Reliability, Experimentation, Security, Human Factors, Standardization, Languages, Theory, Legal Aspects, Verification.

%\terms{Theory, Economics}

%Keywords are your own choice of terms you would like the paper to be indexed by.

%\keywords{Course Allocation, PPAD, A-CEEI, Market Design, Matching}

\section{Introduction}
University classes have limited capacity, and some are more popular than others. This creates an interesting allocation problem.  Imagine that each student has ordered all possible bundles of courses from most desirable to least desirable, and the capacities of classes are known.  What is the best way to allocate class seats to students?  There are several desiderata for a course allocation mechanism: 
\begin{description}
\item[Fairness]  In what sense is the mechanism ``fair''?
\item[Efficiency] Are all possible seats in courses allocated? 
\item[Feasibility]  Are any classes oversubscribed?
\item[Truthfulness] Are students motivated to honestly report their preferences to the mechanism?
\item[Computational efficiency]  Can the allocation be computed from the data in polynomial time?
\end{description}

Competitive Equilibrium from Equal Incomes (CEEI) \citep{Foley67:Resource,Varian74:Equity,Thomson85:Theories} is a venerable mechanism with many attractive properties: In CEEI all agents are allocated the same amount of ``funny money", next they declare their preferences, and then a price equilibrium is found that clears the market.  
The market clearing guarantees efficiency and feasibility. The mechanism has a strong, albeit technical, ex post fairness guarantee that emerges from the notion that agents who miss out on a valuable, competitive item will have extra funny money to spend on other items at equilibrium. Truthfulness is problematic --- as usual with market mechanisms --- even though the problem is mitigated by the large number of agents.   However, CEEI works when the resources to be allocated are divisible and the utilities relatively benign.  It is easy to construct examples in which a CEEI does not exist when preferences are complex or the resources being allocated are not divisible. Indeed, both issues arise in practice in a variety of allocation problems, including shifts to workers, landing slots to airplanes, and our favorite, courses to students. ~\citep{Varian74:Equity,ACEEI_Bud11}. 

It was recently shown in \citet{ACEEI_Bud11} that an approximation to a CEEI solution, called A-CEEI, exists even when the resources are indivisible and agent preferences are arbitrarily complex, as required by the course allocation problems one sees in practice.   The approximate solution guaranteed to exist is approximately fair (in that the agents are given almost the same budget), and approximately efficient and feasible (in that all classes are filled close to capacity, with the possible exception of very unpopular classes).  This result seems to be wonderful news for the class allocation problem.  However, there is a catch:  Budish's proof is non-constructive as it relies on Kakutani's fixed-point theorem.

A heuristic algorithm for solving A-CEEI was introduced in~\citet{Othman10:Finding}. The algorithm is a modified search analogue to the traditional  t\^atonnement process, where the prices of courses that are oversubscribed are increased, and the prices of courses that are undersubscribed are decreased. This heuristic algorithm is currently used by the Wharton School (University of Pennsylvania) to assign their MBA students to courses. It has been documented that the heuristic algorithm often produces much tighter approximations than the theoretical bound; yet, on some instances it fails to find even the guaranteed approximation~\citep[Section 9]{ACEEI_Bud11}.

Thus  A-CEEI is a problem where practical interest motivates theoretical inquiry.  We have a theorem that guarantees the existence of an approximate equilibrium --- the issue is finding it.  Can the heuristic of ~\citet{Othman10:Finding} be replaced by a fast and rigorous algorithm for finding an approximate CEEI?  Or are there complexity obstacles to approximating CEEI?

In this paper, we show that finding the guaranteed approximation to CEEI is an intractable problem: \\

\begin{theorem_again}{\ref{thm:ppad-hard}, informal statement}
The problem of finding an A-CEEI as guaranteed by \citet{ACEEI_Bud11} is $\PPAD$-complete.
\end{theorem_again}\\

We also show an essentially optimal $\NP$-hardness result for determining whether a better approximation exists. \\

\begin{theorem_again}{\ref{thm:np}, informal statement}
It is $\NP$-hard to distinguish between an instance where an exact CEEI exists, and one in which there is no A-CEEI tighter than guaranteed in \citet{ACEEI_Bud11}.
\end{theorem_again}\\

%We begin by describing the mathematical background of A-CEEI.

\section{The Course Allocation Problem}
 %%% NB: This is a DIRECT copy of the text from the AAMAS paper. It should be edited for content and to match notation....

Even though the A-CEEI and the existence theorem in \citep{ACEEI_Bud11} are applicable to a broad range of allocation problems, we shall describe our results in the language of the course allocation problem.

We are given a set of $M$ courses with integer capacities (the supply) $(q_{j})_{j=1}^{M}$, and a set of $N$ students, where each student $i$ has a set $\Psi_{i}\subseteq 2^{M}$
of permissible course bundles, with each bundle containing at most $k \leq M$ courses.  The set $\Psi_{i}$
encodes both scheduling constraints (e.g., courses that meet
at the same time) and any constraints specific to student $i$ (e.g. prerequisites).
%%%% TODO: Utility function is a strict ordering over the set $\Psi_i$, not a real map

% The ordering is denoted by \left(\preccurlyeq_{i}\right)_{i=1}^{N}\right)$

Each student $i$ has a strict ordering over her permissible schedules, denoted by $\preccurlyeq_{i}$.  We allow arbitrarily complex preferences --- in particular,
students may regard courses as substitutes or complements.  More formally:

\begin{defn} \textbf{Course Allocation Problem}
The input to a course allocation problem consists of:
\begin{itemize}
\item For each student $i$ a set of course bundles $\left(\Psi_{i}\right)_{i=1}^{N}$.
\item The students' reported preferences, $\left(\preccurlyeq_{i}\right)_{i=1}^{N}$,
\item The course capacities, $\left( q_j \right)_{j=1}^{M}$, and
\end{itemize}

The output to a course allocation problem consists of:
\begin{itemize}
\item Prices for each course $(p_{j}^{*})_{j=1}^{M}$,
\item Allocations for each student$(x_{i}^{*})_{i=1}^{N}$, and
\item Budgets for each student $\left(b_{i}^{*}\right)_{i=1}^{N}$.
\end{itemize}

%Output,  and budgets $i$ a budget that is a uniform random draw from $[1,1+\beta]$, with $0<\beta\ll\min(\frac{1}{N},\frac{1}{k-1})$.
\end{defn}

How is an allocation evaluated? The \emph{clearing error} of a solution to the allocation problem, is the $\mathcal{L}_2$ norm of the length-$M$ vector of seats oversubscribed in any course, or undersubscribed seats in courses with positive price.

\begin{defn}
The \emph{clearing error} $\alpha$ of an allocation is 

\[ \alpha \equiv \sqrt{\sum_j z_j^2} \]
 Where $z_j$ is given by
\[
z_j  =  \left\{
\begin{array}{ll}
\sum_{i}x_{ij}^{*}-q_{j} & \mbox{ if $p^*_j > 0$}; \\
  \max\left[\left(\sum_{i}x_{ij}^{*}-q_{j}\right),0\right] & \mbox{ if $p^*_j = 0$}.
  \end{array}\right.
\]
\end{defn}

We can now define the notion of {\em approximate} CEEI. 
The quality of approximation is characterized by two parameters: 
$\alpha$, the clearing error (how far is our solution from a true competitive equilibrium?) and
$\beta$, the bound on the difference in budgets (how far from equal are the budgets?). Informally, $\alpha$ can be thought of as the approximation loss on efficiency, and $\beta$ can be thought of as the approximation loss on fairness.

\begin{defn}
An allocation is a \emph{$(\alpha, \beta)$-CEEI} if:
\begin{enumerate}
\item Each student is allocated their most preferred affordable bundle. Formally
\[ \forall i:x_{i}^{*}=\arg\max_{\preccurlyeq_i }
	\left[ x_i\in\Psi_{i} :\sum_{j}x_{ij}p_{j}^{*}\leq b_{i}^{*}\right] \]
\item Total clearing error is at most $\alpha$.
\item Every budget $b^*_i \in [1,1+\beta]$.

\end{enumerate}
\end{defn}

In \citet{ACEEI_Bud11} it is proved that an $(\alpha,\beta)$-approximate CEEI always exists, for some quite favorable (and as we shall see, essentially optimal) values of $\alpha$ and $\beta$:

\begin{thm}
\label{main_budish_theorem}
\citet{ACEEI_Bud11}
For any input preferences, there exists a $(\alpha, \beta)$-CEEI with $\alpha=\sqrt{kM/2}$ and any $\beta > 0$.
\end{thm}

Recall that $k$ is the maximum bundle size. 
$\alpha=\sqrt{kM/2}$ means that, for large number of students and course capacities, the market-clearing error converges to zero quite fast as a fraction of the endowment.   
It is also shown in \citet{ACEEI_Bud11} that the mechanism 
which allocates courses according to such an A-CEEI satisfies attractive criteria of 
%fairness, an intuitive relaxation, and is Pareto efficient outside of the small amount of market-clearing error. 
approximate fairness, approximate truthfulness, and approximate Pareto efficiency.
The reader may consult \citet{ACEEI_Bud11}
for the precise definitions of the economic properties of the A-CEEI mechanism.

\subsection*{Total Functions and PPAD}
Theorem \ref{main_budish_theorem}  is an example of a non-constructive existential result; such theorems are common in mathematics, and are quite often related to economics (recall Nash's theorem, Arrow-Debreu theorem, etc.).  It is often important to determine whether there is a polynomial algorithm for finding the solution guaranteed by such a theorem; computational problems of this nature are called {\em total}, because they correspond to total functions from inputs to solutions.  

In exploring the difficulty of total problems, applying the methodology of $\NP$-completeness is problematic.  
The intuitive reason is that, for example, a reduction from 3SAT relies heavily on the fact that the starting 3SAT instance {\em may be unsatisfiable}.   
Therefore 3SAT cannot be reduced in any meaningful way to a total problem such as  A-CEEI (see Chapter 2 of \cite{AGT} for a discussion of this point).  
$\NP$-completeness does not seem to be an option.  But there is an alternative:  Total problems can often be proved complete for  certain complexity classes between $\P$ and $\NP$.  
For example, during the past decade several game-theoretic problems have been proved complete for the complexity class $\PPAD$, containing difficult problems related to fixed point theorems such as Brouwer's, Nash's, competitive equilibria, and so on 
(\cite{PPAD_Pap94, Nash_winlose, Leontief, leontief_HT07, 2-player_nash_CDT09, Chen_Deng_AD_Fisher, NASH-is-PPAD-hard_DGP09, Reduction_with_same_gadgets_KPRST09, 2dtucker, social_arrow_debreu_CT11, fisher_VY11, arrow_debreu_non-monotone_CPY13}).  

\ignore{
There are several interesting and subtle ways of defining $\PPAD$, but for our purposes it is most convenient to define it as the class of all total problems that are reducible to the following problem, called {\sc 3D-Brouwer,} capturing the complexity of Brouwer's fixed-point theorem.

 \medskip\noindent
 {\sc 3D Brouwer:} Given a function $f$ from the tree-dimensional cube $[0,1]^3$ to itself, and an integer $n$, find an {\em approximate Brouwer fixed point} of $f$.  
 
\medskip 
 It remains to explain how $f$ is described in an instance of the problem {\sc 3D-Brouwer}, and what an approximate fixed point is.  We assume that $f$ is given as a Boolean circuit $C_f$ with $3n$ inputs and two outputs.  The $3n$ Boolean inputs describe a point in the three dimensional cube with coordinates multiples of $2^{-n}$, while the two outputs encode a small displacement vector $D\in \{(\delta,0,0),(0,\delta,0),(0,0,\delta),(-\delta,\delta,-\delta)\}$, such that $f(x)=x+D$.  The function $f$, whose approximate Brouwer fixed point we seek, is then meant to be an interpolation of the function described above.   Finally, by an ``approximate Brouwer fixed point'' we mean something very weak, namely a point $x$ with coordinates multiples of of $2^{\-n}$ such that among there are four different other such vertices at a distance of $2^{-n}$ from it that are displaced by $f$ in all four possible values of $D$.  This problem is shown in \citet{NASH-is-PPAD-hard_DGP09} to be $\PPAD$-complete.
}

There are several interesting and subtle ways of defining $\PPAD$, but for our purposes it is most convenient to define it as the class of all total problems that are reducible to the problem {\sc Gcircuit}, the problem of finding the fixed point of a continuous function specified by a {\em ``generalized circuit''}.  {\sc Gcircuit} is defined in the next section.

%[{\sc $\epsilon$-Gcircuit, \citet{2-player_nash_CDT09}}]
\section{A-CEEI is $\PPAD$-Complete}
\label{sec:A-CEEI_is_PPAD-C}

\begin{thm}
\label{thm:ppad-hard}
%\label{thm:in-ppad}
Computing a $\left(\sqrt{\frac{k M}{2}},\beta\right)$-CEEI is
$\PPAD$-complete, for some polynomially small $\beta>0.$
\end{thm}
The rest of this section is devoted to the proof of this theorem.

\subsubsection*{Membership in $\PPAD$}
We first establish that the problem belongs to the class $\PPAD$; this proof is much harder than usual (see Appendix \ref{sec:in_ppad}).  We follow the steps of the existence proof in \citet{ACEEI_Bud11}, and show that each one can be carried out either in polynomial time, or through a fixed point.  One difficulty is that certain steps of Budish's proof are randomized, and we must be derandomized in polynomial time.

\subsubsection*{The problem {\sc Gcircuit}}
The reduction is from the PPAD-complete problem {\sc Gcircuit}, alluded to in the previous section.

Generalized circuits are similar to the standard algebraic circuits, 
the main difference being that generalized circuits contain cycles,
which allow them to verify fixed points of continuous functions.
Formally,

\begin{defn}[{Generalized circuits, \citet{2-player_nash_CDT09}}]
A {\em generalized circuit} $S$ is a pair $(V,\cal{T})$, where $V$ is a set of nodes and $\cal{T}$ is a
collection of gates. Every gate $T \in \cal{T}$ is a 5-tuple $T = (G, v_1, v_2, v)$, 
in which $G \in \{G_{/2}, G_{\frac{1}{2}}, G_{+}, G_{ - }, G_{<}, G_{\wedge}, G_{\vee}, G_{\neg}\}$%
\footnote{\citet{2-player_nash_CDT09} define slightly different gates, 
whose $\epsilon$-approximation can be simulated by $O(\log 1/\epsilon)$ of our gates: 
$G_{\zeta}$ can be simulated using $G_{\frac{1}{2}}$ and $G_{+}$;
and $G_{\times \zeta}$ and $G_{=}$ can be simulated using $G_{/2}$ and $G_{+}$.}
 is the type of the gate;
$v_1, v_2 \in V \cup \{nil\}$ are the first and second input nodes of the gate;
$v \in V$ is the output node.

The collection $\cal{T}$ of gates must satisfy the following important property:
For every two gates $T = (G, v_1, v_2, v)$ and $T' = (G', v'_1, v'_2, v')$ in $\cal{T}$, 
$v \neq v'$.
\end{defn}

Given a generalized circuit, we are interested in the computational problem of finding an assignment that simultaneously satisfies all the constraints defined by the gates.

\begin{defn} 
Given a generalized circuit $\cal{S} = (V,\cal{T})$, 
we say that an assignment $\mathrm{x} \colon V \rightarrow \mathbb{R}$ {\em $\epsilon$-approximately satisfies} $\cal{S}$ if:

\[
\forall v\in V \;\;\;  0 \leq \mathrm{x}[v] \leq 1 + \epsilon;
\]
and for each gate $T = (G, v_1, v_2, v) \in \cal{T}$ we have that $| \mathrm{x}[v] - f_G(\mathrm{x}[v_1], \mathrm{x}[v_2]) | < \epsilon$, 
where $f_G$  is defined as follows, depending on the type of gate $G$:
\begin{enumerate}
\item HALF: $f_{G_{/2}}\left(x\right)=x/2$
\item VALUE: $f_{G_{\frac{1}{2}}} \equiv\frac{1}{2}$
\item SUM: $f_{G_{+}}\left(x,y\right)=\min\left(x+y,1\right)$
\item DIFF: $f_{G_{-}}\left(x,y\right)=\max\left(x-y,0\right)$
\item LESS: $f_{G_{<}}\left(x,y\right)=\begin{cases}
1 & x>y+\beta\\
0 & y>x+\beta
\end{cases}$
\item AND: $f_{G_{\wedge}}\left(x,y\right)=\begin{cases}
1 & \left(x>\frac{1}{2}+\beta\right)\wedge\left(y>\frac{1}{2}+\beta\right)\\
0 & \left(x<\frac{1}{2}-\beta\right)\vee\left(y<\frac{1}{2}-\beta\right)
\end{cases}$
\item OR: $f_{G_{\vee}}\left(x,y\right)=\begin{cases}
1 & \left(x>\frac{1}{2}+\beta\right)\vee\left(y>\frac{1}{2}+\beta\right)\\
0 & \left(x<\frac{1}{2}-\beta\right)\wedge\left(y<\frac{1}{2}-\beta\right)
\end{cases}$
\item NOT: $f_{G_{\neg}}\left(x\right) = 1-x$
\end{enumerate}
\end{defn}

Given a generalized circuit $\cal{S} = (V,\cal{T})$, 
{\sc $\epsilon$-Gcircuit} is the problem of finding an assignment that $\epsilon$-approximately satisfies it. 
It is shown in \citet{2-player_nash_CDT09} to be $\PPAD$-complete for $\epsilon = {1\over \poly(|V|)}$.

\subsubsection*{Overview of the Reduction}
We shall reduce the {\sc $\epsilon$-Gcircuit} problem to that of finding an $(\alpha,\beta)$-CEEI, with approximation parameters $\alpha = \Theta(N/M)$ and $\epsilon = \beta/2$.
(Note that, by increasing $N$, we can make $\alpha$ arbitrarily large as a function of $M$; in particular, $\alpha > \sqrt{k M/2}$.)
%Similarly to the graphical games of \citet{NASH-is-PPAD-hard_DGP09},

We will construct gadgets (that is, small sets of courses, students, capacities and preferences) for the various types of gates in the generalized circuit.
Each gadget that we construct has one or more dedicated ``input course(s)'', a single ``output course'', and possibly some ``interior courses''.  An output course of one gadget can (and will) be an input to another.
The construction will guarantee that in any A-CEEI the price of
the output course will be approximately equal to the gate applied
to the prices of the input courses.

\subsubsection*{Gate gadgets}

To illustrate what needs to be done, we proceed to construct a gate for the function
$f_{G_{\neg}}\left(x\right)=1-x$; in particular, this implements a logic NOT.
\begin{lemma}
\label{lem:not_gadget}(NOT gadget) 

Let $n_{x}>4\alpha$ and suppose that the economy contains the following
courses:
\begin{itemize}
\item $c_{x}$ (the ``input course'') ;
\item $c_{{1-x}}$ with capacity $q_{{1-x}}=n_{x}/2$ (the
``output course'');
\end{itemize}

and the following set of students:
\begin{itemize}
\item \begin{flushleft}
$n_{x}$ students interested only in the schedule $\left\{ c_{x},c_{{1-x}}\right\} $;
\par\end{flushleft}
\end{itemize}

and suppose further that at most $n_{{1-x}}=n_{x}/4$ other
students are interested in course $c_{{1-x}}$.

Then in any $\left(\alpha,\beta\right)$-CEEI
\[
p_{{1-x}}^{*}\in\left[1-p_{x}^{*},1-p_{x}^{*}+\beta\right]
\]

\end{lemma}
\begin{proof}
Observe that:
\begin{itemize}
\item If $p_{{1-x}}^{*}>1-p_{x}^{*}+\beta$, then none of the $n_{x}$
students will be able to afford the bundle $\left\{ c_{x},c_{{1-x}}\right\} $,
and therefore there will be at most $n_{{1-x}}=n_{x}/4$ students
enrolled in the $c_{{1-x}}$ - much less than the capacity
$n_{x}/2$. Therefore $z_{{1-x}}\geq n_{x}/4$.
\item On the other hand, if $p_{{1-x}}^{*}<1-p_{x}^{*}$, then all
$n_{x}$ students can afford the bundle $\left\{ c_{x},c_{{1-x}}\right\} $
- therefore the class will be overbooked by $n_{x}/2$; thus, $z_{{1-x}}\geq n_{x}/2$.
\end{itemize}
Therefore if $p_{{1-x}}^{*}\notin\left[1-p_{x}^{*},1-p_{x}^{*}+\beta\right]$,
then $\left\Vert z \right\Vert_{2}\geq n_{x}/4>\alpha$ - a contradiction to
$\left(\alpha,\beta\right)$-CEEI.
\end{proof}

Similarly, we construct gadgets that simulate all the gates of the generalized circuit:

\begin{lemma}
\label{lem:additional-gadgets}
Let $n_{x}\geq2^{8}\cdot\alpha$ and suppose that the economy has
courses $c_{x}$ and $c_{y}$. 
Then for any of the gate functions $f_G$ in the definition of {\sc $\epsilon$-Gcircuit}, 
we can add: a course $c_{z}$, and at most $n_{x}$
students interested in each of $c_{x}$ and $c_{y}$, 
such that in any $\left(\alpha,\beta\right)$-CEEI 
$p_{z}^{*}\in\left[f\left(p_{x}^{*},p_{y}^{*}\right)-2\beta,f_G\left(p_{x}^{*},p_{y}^{*}\right)+2\beta\right]$.
\ignore{
\begin{enumerate}
\item HALF: $f\left(x\right)=x/2$
\item DIFF: $f\left(x,y\right)=\max\left(x-y,0\right)$
\item SUM: $f\left(x,y\right)=\min\left(x+y,1\right)$
\item VALUE: $f\equiv\frac{1}{2}$
\item LESS: $f\left(x,y\right)=\begin{cases}
1 & x>y+\beta\\
0 & y>x+\beta
\end{cases}$
\item AND: $f\left(x,y\right)=\begin{cases}
1 & \left(x>\frac{1}{2}+\beta\right)\wedge\left(y>\frac{1}{2}+\beta\right)\\
0 & \left(x<\frac{1}{2}-\beta\right)\vee\left(y<\frac{1}{2}-\beta\right)
\end{cases}$
\item OR: $f\left(x,y\right)=\begin{cases}
1 & \left(x>\frac{1}{2}+\beta\right)\vee\left(y>\frac{1}{2}+\beta\right)\\
0 & \left(x<\frac{1}{2}-\beta\right)\wedge\left(y<\frac{1}{2}-\beta\right)
\end{cases}$
\end{enumerate}
In particular, $p_{z}^{*}\in\left[f\left(p_{x}^{*},p_{y}^{*}\right)-2\beta,f\left(p_{x}^{*},p_{y}^{*}\right)+2\beta\right]$
}

In particular, $p_{z}^{*}$ continue to satisfy the above inequalities in every $\left(\alpha,\beta\right)$-CEEI 
even if up to $n_{z}\leq n_{x}/2^{8}$ additional students (beyond the ones needed in the proof)
are interested in course $c_{z}$.
\end{lemma}

We defer the proof of Lemma \ref{lem:additional-gadgets} to the appendix.

\subsubsection*{Course-size amplification}

So far, we have constructed gadgets that compute all the gates necessary for the
circuit in the reduction from {\sc $\epsilon$-Gcircuit}.
What happens when we try to concatenate them to form a circuit?
Recall the last sentence in the statement of Lemma \ref{lem:additional-gadgets}:
It says that the prices continue to behave like the gate that is simulated, 
as long as there are not too many additional students that try to take the output course.
(If there are more students, 
they may raise the price of the course beyond what we expect.)
In particular, {\em the number of additional students that may want the output course
is smaller than the number of students that want the input course.}

If we concatenated the gadgets without change,
we would need to have larger class sizes as we increase the depth of the simulated circuit.
This increase in class size is exponential in the depth of the circuit.
Things get even worse- since we reduce from {\em generalized} circuits,
our gates form cycles.
If the class size must increase at every gate it would have to be infinite!

To overcome this problem we construct a COPY gadget that preserves
the price from the input course, but is robust to twice as many additional
students:
\begin{lemma}
\label{lem:course-size_amp}
(Course-size amplification gadget) 

Let $n_{x}\geq100\alpha$ and suppose that the economy contains the
following courses:
\begin{itemize}
\item $c_{x}$ (the ``input course'')
\item for $i=1,\dots10$, $c_{i}$ with capacities $q_{i}=0.5\cdot n_{x}$
(``interior courses'');
\item $c_{x'}$ with capacity $q_{x'}$, s.t. $q_{x}\leq q_{x'}\leq4n_{x}$
(``output course'');
\end{itemize}

and the following sets of students:
\begin{itemize}
\item $n_{x}$ students interested in schedules $\left(\left\{ c_{x},c_{i}\right\} \right)_{i=1}^{10}$
(in this order);
\item $n_{i}=0.49\cdot n_{x}$ students ($\forall i$) interested in schedules\\
 $\left(\left\{ c_{x'},c_{i}\right\} ,\left\{ c_{i}\right\} ,\left\{ c_{i+1}\right\} ,\dots,\left\{ c_{10}\right\} \right)$
(in this order);
\end{itemize}

and suppose further that at most $n_{x'}=2n_{x}$ other students are
interested in course $c_{x'}$.

Then in any $\left(\alpha,\beta\right)$-CEEI
\[
p_{x'}^{*}\in\left[p_{x}^{*}-\beta,p_{x}^{*}+\beta\right]
\]

\end{lemma}
In particular, notice that the price of $c_{x'}$ is guaranteed to
approximate the price of $c_{x}$, even in the presence of additional
$n_{x'}=2n_{x}$ students - twice as many students as we added to
$c_{x}$.
\begin{proof}
We start by proving that all the $c_{i}$'s simulate NOT gadgets simultaneously,
i.e. for every $i$ and every $\left(\alpha,\beta\right)$-CEEI,
$p_{i}^{*}\in\left[1-p_{x}^{*},1-p_{x}^{*}+\beta\right]$.
\begin{itemize}
\item If $p_{i}^{*}>1-p_{x}^{*}+\beta$, assume wlog that it is the first
such $i$, i.e. $p_{j}^{*}\leq1-p_{x}^{*}+\beta<p_{i}^{*}$ for every
$j<i$.

None of the $n_{x}$ students can afford buying both $c_{x}$ and
$c_{i}$. Furthermore, for every $j<i$, none of the $n_{j}$ students
will prefer $c_{i}$ over $c_{j}$. Therefore at most $n_{i}$ students
will take this course: $z_{i}^{*}\geq0.01n_{x}$.

\item If, on the other hand, $p_{i}^{*}<1-p_{x}^{*}$, then all $n_{x}$
students will buy course $c_{i}$ or some previous course $c_{j}$
(for $j\leq i$); additionally for every $j\leq i$, each of the $n_{j}$
corresponding students will buy some course $c_{k}$ for $j\leq k\leq i$.
Therefore the total overbooking of classes $1,\dots,i$ will be at
least $\sum_{j\leq i}z_{j}^{*}\geq n_{x}\cdot\left(1-0.01i\right)$
- a contradiction to $\left(\alpha,\beta\right)$-CEEI.
\end{itemize}
Now that we established that $p_{i}^{*}\in\left[1-p_{x}^{*},1-p_{x}^{*}+\beta\right]$,
we shall prove the main claim, i.e. that $p_{x'}^{*}\in\left[p_{x}^{*}-\beta,p_{x}^{*}+\beta\right]$.
\begin{itemize}
\item If $p_{x'}^{*}>p_{x}^{*}+\beta$, then none of the $n_{i}$ students,
for any $n_{i}$, can afford buying both $c_{x'}$ and $c_{i}$. Therefore,
even in the presence of additional $n_{x'}=2n_{x}$ students who want
to take $c_{x'}$, the class will be undersubscribed by $z_{x'}^{*}\geq q_{x'}-n_{x'}=2n_{x}$
\item If $p_{x'}^{*}<x+\beta$, then all $n_{i}$ students, for each $i$,
can afford to buy their top schedule - both $\left\{ c_{i},c_{x'}\right\} $.
Therefore $c_{x'}$ will be oversubscribed by at least $z_{x'}^{*}\geq0.9\cdot n_{x}$
- a contradiction to $\left(\alpha,\beta\right)$-CEEI.
\end{itemize}
\end{proof}

Finally, given an instance of {\sc $\epsilon$-Gcircuit}, 
we can use  the gadgets we constructed in Lemmata \ref{lem:not_gadget}-\ref{lem:course-size_amp}
%together with the techniques of \citet{NASH-is-PPAD-hard_DGP09, Reduction_with_same_gadgets_KPRST09},
to construct an instance of $(\alpha,\beta)$-CEEI  that simulates the generalized circuit.
%as well the averaging gadget of \citet{NASH-is-PPAD-hard_DGP09}. 
%(See \citet[Theorem 3.3]{Reduction_with_same_gadgets_KPRST09} for another proof that relies in the same way on the proof in \citet{NASH-is-PPAD-hard_DGP09}.)

\qed\\

%As we discussed in the previous section, 
%$\PPAD$-hardness is a slightly weaker notion than $\NP$-hardness.
%Since we proved that Budish's approximate CEEI belongs to $\PPAD$,
%we cannot hope (without very strong complexity implications),
%to prove that it is $\NP$-hard.
%In the next section we prove the strongest 
%$\NP$-hardness of approximation result we could hope for:
%Finding any stronger approximation of CEEI is already $\NP$-hard.

\section{$\NP$ hardness}
\citet{ACEEI_Bud11} shows that his existence theorem is tight, that is,
there exist economies in which it is impossible to achieve less than 
$\Omega\left(\sqrt{k M}\right)$ market clearing error.
One may hope that on instances encountered in practice, 
a better approximation may be possible, and finding it may not be prohibitively hard.
We next show that even in {\em economies that admit an exact CEEI},
it is $\NP$-hard to find even a {\em constant factor improvement} over the 
$\Omega\left(\sqrt{k M}\right)$ bound.

\begin{thm}
\label{thm:np}
It is $\NP$-hard to distinguish between an economy that has an exact
CEEI, and an economy that does not have a $\left( \Omega\left(\sqrt{N+M}\right) ,\beta\right)$-CEEI
for any $0\leq\beta<1$.
\end{thm}

In particular, since our reduction uses a constant $k$, it means
that it is $\NP$-complete to find an $\left(\Omega\left(\sqrt{k M}\right),\beta\right)$-CEEI
--- an approximation factor smaller only by a multiplicative constant than the approximation guaranteed by the existence theorem of \citet{ACEEI_Bud11}.

\subsubsection*{Comparison to Theorem  \ref{thm:ppad-hard}}
Theorem  \ref{thm:ppad-hard} is in some sense stronger than Theorem  \ref{thm:np}
in that it applies to a larger market clearing error.
In turn, Theorem  \ref{thm:np} is stronger in two ways:
(1) it gives $\NP$-hardness, as opposed to $\PPAD$-hardness; and
(2) it applies to any $0 \leq \beta < 1$, as opposed to a polynomially small $\beta$.

%The proof of Theorem \ref{thm:np} is given in the Appendix.

\subsection{Proof}
We reduce from 3SAT-5, i.e., a SAT instance in which every clause contains
exactly 3 variables, and each variable appears in exactly 5 clauses.
\citet{sparse_SAT_inapproximability} proved that it is $\NP$-hard
to distinguish between a satisfiable 3SAT-5 instance, and a 3SAT-5
instance where at most $1-\epsilon$ can be satisfied, for some $\epsilon>0$%
\footnote{In fact, an equivalent result for 3SAT-$B$ for any constant $B$
would suffice for our techniques. Hardness of approximation with perfect
completeness for 3SAT-$B$ was proven by \citet{sparse_SAT_inapproximability_PY,sparse_SAT_inapproximability_AS,sparse_SAT_inapproximability_ALMSS}.%
}. 

Given a 3SAT-5 formula, we construct a gadget for each variable and
each clause. The gadgets are constructed so that for any assignment
that completely satisfies the formula there exists an exact CEEI
in the economy. 

Furthermore, given an approximate CEEI for the economy which exactly
clears the courses in a {\em subset} of the gadgets, one can recover an
assignment for the 3SAT-5 formula that satisfies all the clauses corresponding
to the same subset. Informally, this means that for every clause that we are
unable to satisfy in the 3SAT-5 formula, there must be a deviation
from exact market clearing in the gadget corresponding to either that
clause, or one of its variables.

Because we use a sparse 3SAT, each deviation from market clearing
can affect at most 5 clauses. Each variable gadget uses $13$ courses,
and each clause gadget uses only $1$ more. For an instance with $n$
clauses and $\frac{3}{5}n$ variables, we have exactly $M=\frac{44}{5}n$
courses. Finally, if $\epsilon n$ of the clauses are unsatisfied,
then the market clearing error must be at least $\sqrt{\frac{1}{5}\cdot\epsilon n}=\sqrt{\frac{\epsilon}{44}\cdot M}$.
Since $N<M$ and $\epsilon>0$ is a constant, 
we get $\NP$-hardness with $\alpha = \Omega\left(\sqrt{M+N}\right)$.

%%%%%%%%%%%%%%%%%%%%%%%%%%%%%%%%%%%%%%%%%%%%
\begin{figure}
\footnotesize
\caption{Example gadgets}

\subfloat[Variable gadget, student $s_{T}$]
{

\begin{minipage}[t]{1\columnwidth}%
\begin{center}
\tikzstyle{every node}=[font=\scriptsize]
\begin{tikzpicture}[scale = 0.6]
%
% X
\coordinate (D_L1) at (0,0); \coordinate (D_R1) at (3.2,0); \coordinate (D_C1) at ($0.5*(D_L1) + 0.5*(D_R1)$);
\coordinate (O_T1) at ($(D_L1)+(0,4)$); \coordinate (O_F1) at ($(D_R1)+(O_T1)-(D_L1)$);
\coordinate (O_T2) at ($(O_T1) + (-0.5,-0.9)$); \coordinate (O_T3) at ($(O_T1) + (0,-0.9)$); \coordinate (O_T4) at ($(O_T1) + (0.5,-0.9)$); \coordinate (O_T5) at ($(O_T1) + (0,-1.4)$); \coordinate (O_F2) at ($(O_F1) + (-0.5,-0.9)$); \coordinate (O_F3) at ($(O_F1) + (0,-0.9)$); \coordinate (O_F4) at ($(O_F1) + (0.5,-0.9)$); \coordinate (O_F5) at ($(O_F1) + (0,-1.4)$);
\draw [fill=black] (D_L1) circle [radius=0.1]; \draw [fill=black] (D_C1) circle [radius=0.1]; \draw [fill=black] (D_R1) circle [radius=0.1]; \draw [fill=black] (O_T1) circle [radius=0.1]; \draw [fill=black] (O_F1) circle [radius=0.1];
\draw [fill=black] (O_T2) circle [radius=0.1]; \draw [fill=black] (O_T3) circle [radius=0.1]; \draw [fill=black] (O_T4) circle [radius=0.1]; \draw [fill=black] (O_T5) circle [radius=0.1]; \draw [fill=black] (O_F2) circle [radius=0.1]; \draw [fill=black] (O_F3) circle [radius=0.1]; \draw [fill=black] (O_F4) circle [radius=0.1]; \draw [fill=black] (O_F5) circle [radius=0.1];
\node [below left] at (D_L1) {$D_L$}; \node [below ] at (D_C1) {$D_C$}; \node [below right] at (D_R1) {$D_R$};
\node [below] at (O_T1) {$X_T^1$}; \node [below] at (O_F1) {$X_F^1$};
\draw [loosely dashed] ($0.5*(D_L1)+0.5*(D_C1)$) ellipse (2 and 1);
\draw [loosely dashed] ($0.5*(D_L1)+0.5*(O_T1)$) ellipse (1 and 2.5);
\draw [loosely dashed] (D_R1) circle (0.5);
\end{tikzpicture} 
\end{center}
The courses (dots) and bundles (ellipses) that interest student $s_{T}$
of the variable gadget. In particular, note that $s_{T}$ may take
the courses corresponding to the assignment $X=\mbox{True}$.%
\end{minipage}
}

\vspace{0.3cm}

\subfloat[Variable gadget, student $s_{F}$]{
\begin{minipage}[t]{1\columnwidth}%
\begin{center}
\tikzstyle{every node}=[font=\scriptsize]
\begin{tikzpicture}[scale = 0.6]
% X
\coordinate (D_L1) at (0,0); \coordinate (D_R1) at (3.2,0); \coordinate (D_C1) at ($0.5*(D_L1) + 0.5*(D_R1)$);
\coordinate (O_T1) at ($(D_L1)+(0,4)$); \coordinate (O_F1) at ($(D_R1)+(O_T1)-(D_L1)$);
\coordinate (O_T2) at ($(O_T1) + (-0.5,-0.9)$); \coordinate (O_T3) at ($(O_T1) + (0,-0.9)$); \coordinate (O_T4) at ($(O_T1) + (0.5,-0.9)$); \coordinate (O_T5) at ($(O_T1) + (0,-1.4)$); \coordinate (O_F2) at ($(O_F1) + (-0.5,-0.9)$); \coordinate (O_F3) at ($(O_F1) + (0,-0.9)$); \coordinate (O_F4) at ($(O_F1) + (0.5,-0.9)$); \coordinate (O_F5) at ($(O_F1) + (0,-1.4)$);
\draw [fill=black] (D_L1) circle [radius=0.1]; \draw [fill=black] (D_C1) circle [radius=0.1]; \draw [fill=black] (D_R1) circle [radius=0.1]; \draw [fill=black] (O_T1) circle [radius=0.1]; \draw [fill=black] (O_F1) circle [radius=0.1];
\draw [fill=black] (O_T2) circle [radius=0.1]; \draw [fill=black] (O_T3) circle [radius=0.1]; \draw [fill=black] (O_T4) circle [radius=0.1]; \draw [fill=black] (O_T5) circle [radius=0.1]; \draw [fill=black] (O_F2) circle [radius=0.1]; \draw [fill=black] (O_F3) circle [radius=0.1]; \draw [fill=black] (O_F4) circle [radius=0.1]; \draw [fill=black] (O_F5) circle [radius=0.1];
\node [below left] at (D_L1) {$D_L$}; \node [below ] at (D_C1) {$D_C$}; \node [below right] at (D_R1) {$D_R$};
\node [below] at (O_T1) {$X_T^1$}; \node [below] at (O_F1) {$X_F^1$};
\draw [dotted]  ($0.5*(D_R1)+0.5*(D_C1)$) ellipse (2 and 1);
\draw [dotted]  ($0.5*(D_R1)+0.5*(O_F1)$) ellipse (1 and 2.5);
\draw [dotted]  (D_L1) circle (0.5); 
\end{tikzpicture}

\par\end{center}

The courses (dots) and bundles (ellipses) that interest student $s_{F}$
of the variable gadget. In particular, note that $s_{T}$ may take
the courses corresponding to the assignment $X=\mbox{False}$.%
\end{minipage}
}

\vspace{0.3cm}

\subfloat[Clause gadget]{

\begin{minipage}[t]{1\columnwidth}%
\begin{center}
\begin{tikzpicture}[scale = 0.6]
% X
\coordinate (D_L1) at (0,0); \coordinate (D_R1) at (3.2,0); \coordinate (D_C1) at ($0.5*(D_L1) + 0.5*(D_R1)$);
\coordinate (O_T1) at ($(D_L1)+(0,4)$); \coordinate (O_F1) at ($(D_R1)+(O_T1)-(D_L1)$);
\coordinate (O_T2) at ($(O_T1) + (-0.5,-0.9)$); \coordinate (O_T3) at ($(O_T1) + (0,-0.9)$); \coordinate (O_T4) at ($(O_T1) + (0.5,-0.9)$); \coordinate (O_T5) at ($(O_T1) + (0,-1.4)$); \coordinate (O_F2) at ($(O_F1) + (-0.5,-0.9)$); \coordinate (O_F3) at ($(O_F1) + (0,-0.9)$); \coordinate (O_F4) at ($(O_F1) + (0.5,-0.9)$); \coordinate (O_F5) at ($(O_F1) + (0,-1.4)$);
\draw [fill=black] (O_T1) circle [radius=0.1]; 
\draw [fill=black] (O_F1) circle [radius=0.1]; 
\node [below] at (O_T1) {$X_T^1$}; \node [below] at (O_F1) {$X_F^1$};
% clause
\coordinate (X_T) at (O_T1); \coordinate (X_F) at (O_F1); \coordinate (Y_T) at ($(X_F) + 0.5*(X_F) - 0.5*(X_T)$); \coordinate (Y_F) at ($(Y_T) + 0.3*(X_F) - 0.3*(X_T)$); \coordinate (Z_T) at ($(Y_F) + 0.5*(X_F) - 0.5*(X_T)$); \coordinate (Z_F) at ($(Z_T) + 0.3*(X_F) - 0.3*(X_T)$);
\coordinate (D) at ($0.5*(X_T) + 0.5*(Z_F) + (0, 1.5)$);
\draw [fill=black] (Y_T) circle [radius=0.1]; \draw [fill=black] (Y_F) circle [radius=0.1]; \draw [fill=black] (Z_T) circle [radius=0.1]; \draw [fill=black] (Z_F) circle [radius=0.1]; \draw [fill=black] (D) circle [radius=0.1];
\node [below ] at (Y_T) {$Y_T$}; \node [below ] at (Y_F) {$Y_F$};
\node [below ] at (Z_T) {$Z_T$}; \node [below ] at (Z_F) {$Z_F$};
\node [above left] at (D) {$D$};
\draw [rounded corners=6, solid]  	($(X_F) + (-0.7, 0.7)$) -- ($(X_F) + (-0.7, -0.7)$) -- ($(X_F) + (0.7, -0.7)$) -- ($(X_F) + (0.7, 0.7)$) -- 	($(Y_F) + (-0.7, 0.7)$) -- ($(Y_F) + (-0.7, -0.7)$) -- ($(Y_F) + (0.7, -0.7)$) -- ($(Y_F) + (0.7, 0.7)$) -- 	($(Z_F) + (-0.7, 0.7)$) -- ($(Z_F) + (-0.7, -0.7)$) -- ($(Z_F) + (0.7, -0.7)$) -- ($(Z_F) + (0.7, 1.4)$) -- 	($(D) + (0.7, -0.7)$) -- ($(D) + (0.7, 0.7)$) -- ($(D) + (-0.7, 0.7)$) -- cycle;
\end{tikzpicture} 
\end{center}

The courses (dots) and bundle (shape) corresponding to the assignment  $\left(X,Y,Z\right)=\left(\mbox{True},\mbox{True},\mbox{True}\right)$,
one of the seven bundles that may be picked by the clause student.%
\end{minipage}}

\end{figure}
\normalsize

\begin{figure}[t]
\caption{Putting the gadgets together}
%
%
%\vspace{0.7cm}
%
\begin{minipage}[t]{1\columnwidth}%
\begin{tikzpicture}[scale = 0.7]
\tikzstyle{every node}=[font=\scriptsize]
% X
\coordinate (D_L1) at (0,0); \coordinate (D_R1) at (3.2,0); \coordinate (D_C1) at ($0.5*(D_L1) + 0.5*(D_R1)$);
\coordinate (O_T1) at ($(D_L1)+(0,4)$); \coordinate (O_F1) at ($(D_R1)+(O_T1)-(D_L1)$);
\coordinate (O_T2) at ($(O_T1) + (-0.5,-0.9)$); \coordinate (O_T3) at ($(O_T1) + (0,-0.9)$); \coordinate (O_T4) at ($(O_T1) + (0.5,-0.9)$); \coordinate (O_T5) at ($(O_T1) + (0,-1.4)$); \coordinate (O_F2) at ($(O_F1) + (-0.5,-0.9)$); \coordinate (O_F3) at ($(O_F1) + (0,-0.9)$); \coordinate (O_F4) at ($(O_F1) + (0.5,-0.9)$); \coordinate (O_F5) at ($(O_F1) + (0,-1.4)$);
\draw [fill=black] (D_L1) circle [radius=0.1]; \draw [fill=black] (D_C1) circle [radius=0.1]; \draw [fill=black] (D_R1) circle [radius=0.1]; \draw [fill=black] (O_T1) circle [radius=0.1]; \draw [fill=black] (O_F1) circle [radius=0.1];
\draw [fill=black] (O_T2) circle [radius=0.1]; \draw [fill=black] (O_T3) circle [radius=0.1]; \draw [fill=black] (O_T4) circle [radius=0.1]; \draw [fill=black] (O_T5) circle [radius=0.1]; \draw [fill=black] (O_F2) circle [radius=0.1]; \draw [fill=black] (O_F3) circle [radius=0.1]; \draw [fill=black] (O_F4) circle [radius=0.1]; \draw [fill=black] (O_F5) circle [radius=0.1];
\node [below left] at (D_L1) {$D_L$}; \node [below ] at (D_C1) {$D_C$}; \node [below right] at (D_R1) {$D_R$};
\node [below] at (O_T1) {$X_T^1$}; \node [below] at (O_F1) {$X_F^1$};
\draw [loosely dashed] ($0.5*(D_L1)+0.5*(D_C1)$) ellipse (2 and 1); \draw [dotted]  ($0.5*(D_R1)+0.5*(D_C1)$) ellipse (2 and 1);
\draw [loosely dashed] ($0.5*(D_L1)+0.5*(O_T1)$) ellipse (1 and 2.5); \draw [dotted]  ($0.5*(D_R1)+0.5*(O_F1)$) ellipse (1 and 2.5);
\draw [dotted]  (D_L1) circle (0.5); \draw [loosely dashed] (D_R1) circle (0.5);
% clause
\coordinate (X_T) at (O_T1); \coordinate (X_F) at (O_F1); \coordinate (Y_T) at ($(X_F) + 0.5*(X_F) - 0.5*(X_T)$); \coordinate (Y_F) at ($(Y_T) + 0.3*(X_F) - 0.3*(X_T)$); \coordinate (Z_T) at ($(Y_F) + 0.5*(X_F) - 0.5*(X_T)$); \coordinate (Z_F) at ($(Z_T) + 0.3*(X_F) - 0.3*(X_T)$);
\coordinate (D) at ($0.5*(X_T) + 0.5*(Z_F) + (0, 1.5)$);
\draw [fill=black] (Y_T) circle [radius=0.1]; \draw [fill=black] (Y_F) circle [radius=0.1]; \draw [fill=black] (Z_T) circle [radius=0.1]; \draw [fill=black] (Z_F) circle [radius=0.1]; \draw [fill=black] (D) circle [radius=0.1];
\node [below ] at (Y_T) {$Y_T$}; \node [below ] at (Y_F) {$Y_F$};
\node [below ] at (Z_T) {$Z_T$}; \node [below ] at (Z_F) {$Z_F$};
\node [above left] at (D) {$D$};
\draw [rounded corners=11, solid]  	($(X_F) + (-0.7, 0.7)$) -- ($(X_F) + (-0.7, -0.7)$) -- ($(X_F) + (0.7, -0.7)$) -- ($(X_F) + (0.7, 0.7)$) -- 	($(Y_F) + (-0.7, 0.7)$) -- ($(Y_F) + (-0.7, -0.7)$) -- ($(Y_F) + (0.7, -0.7)$) -- ($(Y_F) + (0.7, 0.7)$) -- 	($(Z_F) + (-0.7, 0.7)$) -- ($(Z_F) + (-0.7, -0.7)$) -- ($(Z_F) + (0.7, -0.7)$) -- ($(Z_F) + (0.7, 1.4)$) -- 	($(D) + (0.7, -0.7)$) -- ($(D) + (0.7, 0.7)$) -- ($(D) + (-0.7, 0.7)$) -- cycle;
\end{tikzpicture} %
\end{minipage}

\vspace{0.3cm}
\footnotesize

Some of the courses (dots) and bundles (shapes) of an economy simulating
the formula $\Psi = X\vee\neg Y\vee Z$. 

From bottom to top: the bundles that interest student $s_{T}$ (dashed)
and $s_{F}$ (dotted) of the variable gadget for variable $X$; and
the bundle (solid) corresponding to the assignment $\left(X,Y,Z\right)=\left(\mbox{True},\mbox{True},\mbox{True}\right)$,
one of the seven bundles that may be picked by the clause student.
\end{figure}
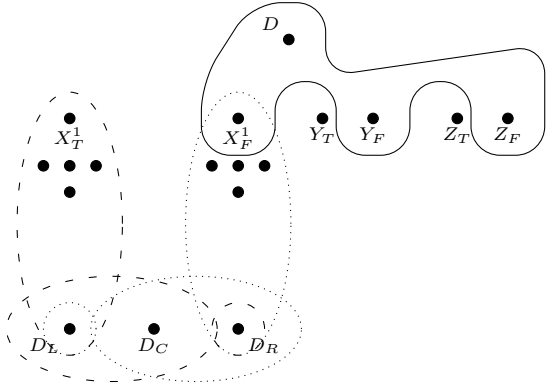

%%%%%%%%%%%%%%%%%%%%%%%%%%%%%%%%%%%%%%%%%%%%

\subsubsection*{Variable gadget}

For each variable $x_{i}$, we have a variable gadget that forces
a consistent assignment to $x_{i}$. 
The gadget contains $5$ pairs of
``output courses'' $O_{T}^{j},O_{F}^{j}$; each of these pairs is
also part of the ``input courses'' of a clause gadget.
Additionally, the gadget has three
inner courses: $D_{L},D_{C},D_{R}$. The gadget also has two students:
$s_{T}$ has preference list: $\left\{ D_{L},D_{C}\right\} $, $\left\{ D_{L},O_{T}^{1},\dots O_{T}^{5}\right\} $,
$\left\{ D_{R}\right\} $; and $s_{F}$ has preference list: $\left\{ D_{R},D_{C}\right\} $,
$\left\{ D_{R},O_{F}^{1},\dots,O_{F}^{5}\right\} $, $\left\{ D_{L}\right\} $.

\begin{itemize}
	\item \textbf{Soundness: }It is easy to see that, in any CEEI, $x_{i}$ cannot
	be assigned more than one value: otherwise neither student will be
	assigned $D_{C}$; yet if $D_{C}$ has price zero, then both students
	would prefer the respective bundles that contain it.

	If, on the other hand, neither $O_{T}^{j}$ nor $O_{F}^{j}$ is assigned,
	we must again have a nonzero market clearing error for the courses in this gadget:
	\begin{itemize}
		\item If all the inner courses have price zero, then $D_{C}$ will be over
		demanded; 

		%\begin{itemize}
		\item If $D_{L}$ and $D_{R}$ have price zero, then under any assignment
		either one of the three will be over demanded, or $D_{C}$ will be
		under demanded;

		%\begin{itemize}
		\item If $p\left(D_{C}\right)=0$, $p\left(D_{L}\right)>0$, and, wlog,
		$p\left(D_{L}\right)\geq p\left(D_{R}\right)$, then either $D_{C}$
		will be over demanded, or $D_{L}$ will be under demanded; 

		\item Finally, since $\beta<1$, if $D_{C}$ has nonzero price, then either
		it is under demanded, or one of the three inner courses must be over
		demanded.
	\end{itemize}

	\item \textbf{Completeness:} For an assignment with $x_{i}=\mbox{True}$,
	let the prices of $O_{T}^{j},O_{F}^{j},D_{L},D_{C},D_{R}$ be $\frac{1}{6},0,\frac{1}{6},1,0$,
	respectively. Under these prices, student $s_{T}$ will prefer bundle
	$\left\{ D_{L},O_{T}^{1},\dots O_{T}^{4}\right\}$%
		\footnote{Recall that in the completeness we show the existence of {\em exact} CEEI, 
				 so all the budgets are exactly $1$.}, 
	while student $S_{F}$ will choose bundle $\left\{ D_{R},D_{C}\right\} $.
\end{itemize}

\subsubsection*{Clause gadget}

For each clause containing variables $\left\{ X,Y,Z\right\} $, consider
seven courses: six input courses $X_{T},X_{F},Y_{T},Y_{F},Z_{T},Z_{F}$
(where each pair is the output courses of a variable gadget), and
a single ``budget diluting'' course $D$. We also have a single
gadget student, who is interested in any of the seven bundles corresponding
to a satisfying assignment.

For example if the clause is $\left(X\vee\neg Y\vee Z\right)$, the
gadget student would be interested in the bundles: $\left\{ X_{F},Y_{F},Z_{F},D\right\} $,
$\left\{ X_{F},Y_{F},Z_{T},D\right\} $, $\left\{ X_{F},Y_{T},Z_{F},D\right\} $,
$\left\{ X_{F},Y_{T},Z_{T},D\right\} $, $\left\{ X_{T},Y_{F},Z_{F},D\right\} $,
$\left\{ X_{F},Y_{T},Z_{F},D\right\} $, $\left\{ X_{T},Y_{T},Z_{T},D\right\} $.
In particular, the student is not interested in the bundle $\left\{ X_{T},Y_{F},Z_{T},D\right\} $,
which corresponds to assigning $\left(X=\mbox{False},\, Y=\mbox{True},\, Z=\mbox{False}\right)$
\begin{itemize}
	\item \textbf{Soundness:} Observe that the variable gadgets students are
	assigned courses $X_{a}$, $Y_{b}$, and $Z_{c}$, then in any exact
	CEEI, the clause gadget student must be assigned the bundle $\left\{ X_{\neg a},Y_{\neg b},Z_{\neg c},D\right\} $.

	\item \textbf{Completeness:} Suppose that the variable gadgets students
	are assigned courses $X_{a}$, $Y_{b}$, and $Z_{c}$, each with price
	at least $\frac{1}{6}$, while courses $X_{\neg a}$, $Y_{\neg b}$,
	and $Z_{\neg c}$ are all unassigned. Then if we set the price of
	$D$ to be $1$, the only affordable bundle for the clause gadget
	student is indeed $\left\{ X_{\neg a},Y_{\neg b},Z_{\neg c},D\right\} $.
\end{itemize}

%\end{proof}

% <<< ADD INTRO TEXT HERE, NEED A STATEMENT OF THE THEOREM TOO >>>

\section{Discussion}

In this work we classified the computational complexity of finding an approximate CEEI 
as a function of the precision parameter $\alpha$ of the approximation, the market clearing error. We showed that finding $(\alpha, \beta)$-CEEI is $\PPAD$-complete 
when  $\alpha$ is large enough to guarantee existence,
while finding a better approximation to CEEI is $\NP$-complete.

One potential way around these intractability results could be to restrict the input language of preferences. This has been a fruitful line of research in combinatorial auctions \citep{nisan_book_chapter, sandholm_book_chapter}. However, in contrast to that  space, we do not anticipate limiting language complexity in the course allocation problem to be fruitful either in theory or in practice. Recall that the student preferences used in the $\PPAD$-hardness proof are already very simple. Furthermore, in practice there are significant inherent complexities in students' preferences: for example, courses meeting at the same time and courses with multiple sections.

Despite the negative results shown in this paper, a heuristic search algorithm exists that finds practical solutions to A-CEEI. Interestingly, in both laboratory experiments as well as real course allocation problems, this heuristic often finds solutions that are an order of magnitude better than the theoretical $\sqrt{\frac{k M}{2}}$ guarantee on the clearing error~\citep{Othman10:Finding} --- a performance which we have shown NP-hard to guarantee.  Once again we are faced with a familiar conundrum:  What are the characteristics of the instances appearing in practice that enable this favorable performance?  And how can one develop a rigorous fast algorithm for them?

%\begin{itemize}
%\item There must exist input preferences on which the heuristic search algorithm will perform poorly. What are these inputs? Are they simple to detect? Are there modifications we can make to the preference elicitation or the assigning of budgets to subvert them?

%\item What kind of input distributions yield better approximations? For example, in the $\PPAD$-hardness proof, a large numbers of students have the same preferences.  Does the problem become easier if there is sufficient diversity of preferences?
%\end{itemize}

%Finally, many economic mechanisms rely on fixed-point theorems; in the postwar formalization of economic theory the fixed-point arguments of Brouwer and Kakutani became indispensable tools for the mathematical economist. It is telling that essentially every $\PPAD$-complete problem comes from economic science. This suggests opportunity for a more nuanced approach to mechanism design that takes into account computational complexity as a first principle. Put another way, can we find an alternative economic mechanism that achieves many or all of the advantageous properties of A- CEEI, but does not rely on the computation of fixed points or other intractable problems?

\bibliographystyle{plainnat}
\bibliography{../A-CEEI}

\appendix

\section{A-CEEI $\in \PPAD$}
\label{sec:in_ppad}

%\begin{theorem_again}{\ref{thm:in-ppad}
We show that computing a $\left(\frac{\sqrt{\sigma M}}{2},\beta\right)$-CEEI is
in $\PPAD$, for $\sigma=\min \{2k,M\}$.
%\end{theorem_again}

\begin{rem}
We assume that the student preferences $\left(\succsim_{i}\right)$ are given in the form of
an ordered list of all the bundles in $\Psi_{i}$ (i.e., all the bundles
that student $i$ prefers over the empty bundle). In particular, we
assume that the total number of permissible bundles is polynomial.
\end{rem}

\begin{rem}
In fact, we prove that the following, slightly more general problem, is in $\PPAD$: 
Given any $\beta,\epsilon>0$ and initial approximate-budgets vector $\mathbf{b}\in\left[1,1+\beta\right]^{N}$,
find a $\left(\frac{\sqrt{\sigma M}}{2},\beta\right)$-CEEI with budgets $\mathbf{b^*}$
such that $ |b_i - b^*_i| < \epsilon $ for every $i$.
\end{rem}

Our proof will follow the steps of the existence proof by \citet{ACEEI_Bud11}.
We will use the power of $\PPAD$ to solve the Kakutani problem, and derandomize
the other nonconstructive ingredients.

\subsection{Preliminaries}

Our algorithm receives as input an economy $\left(\left(q_{j}\right)_{j=1}^{M},\left(\Psi_{i}\right)_{i=1}^{N},\left(\succsim_{i}\right)_{i=1}^{N}\right)$,
parameters $\beta,\epsilon>0$, and an initial approximate-budgets vector $\mathbf{b}\in\left[1,1+\beta\right]^{N}$. We denote $\bar{\beta} = \min\{\beta,\epsilon\} / 2$.

We will consider $M$-dimensional price vectors in ${\cal P}=\left[0,1+\beta+\epsilon\right]^{M}$.
In order to define a price adjustment function, we consider an enlargement
 $\cal{\tilde{P}}=\left[-1,2+\beta+\epsilon\right]^{M}$, as well
as a truncation function $t:\cal{\tilde{P}}\rightarrow{\cal P}$. 

For each student $i$, we denote her demand at prices $\mathbf{\tilde{p}}$
with budget $b_{i}$ by 
\[
d_{i}\left(\mathbf{\tilde{p}},b_{i}\right)={\max}_{\left(\succsim_{i}\right)}\left\{ x'\in \Psi_i \colon\mathbf{\tilde{p}}\cdot x'\leq b_{i}\right\} 
\]
Given the total demand of all the students, we can define the excess
demand to be:
\[
\mathbf{z}\left(\mathbf{\tilde{p}},\mathbf{b}\right)=\sum_{i=1}^{N}d_{i}\left(\mathbf{\tilde{p}},b_{i}\right)-\mathbf{q}
\]

A key ingredient to the analysis is the budget-constraint hyperplanes. These are the hyperplanes in price space along which a student can exactly afford a specific bundle.
For each student $i$ and bundle $x$, the corresponding \emph{budget-constraint
hyperplane} is defined as $H\left(i,x\right)=\left\{ \mathbf{\tilde{p}}\in\cal{{P}}\colon\mathbf{\tilde{p}}\cdot x=b_{i}\right\} $.

\subsection{Deterministically finding a ``general position'' perturbation (step
1)}

It is convenient to assume that the budget-constraint hyperplanes
are in ``general position'', i.e. there is no point $\mathbf{\tilde{p}}\in\cal{{P}}$
at which any subset of linearly dependent budget-constraint hyperplanes
intersect (in particular, no more than $M$ hyperplanes intersect
at any point). In the existence proof, this is achieved by assigning
a small random reverse tax $\tau_{i,x}\in\left(-\epsilon,\epsilon\right)$,
for each student $i$ and bundle $x$; $i$'s modified cost for bundle
$x$ at prices $\mathbf{\tilde{p}}$ becomes $\mathbf{\tilde{p}}\cdot x-\tau_{i,x}$.
Given taxes $\mathbf{\tau}=\left(\tau_{i,x}\right)_{i\in{\cal S},x\in\Psi_{i}}$,
we redefine $d_{i}\left(\mathbf{\tilde{p}},b_{i},\tau_{i}\right)$,
$\mathbf{z}\left(\mathbf{\tilde{p}},\mathbf{b},\mathbf{\tau}\right)$,
and $H\left(i,x,\tau_{i,x}\right)$ analogously. 

In this section, we show how to deterministically choose these taxes.

\begin{lemma}
\label{lem:general_position}There exists a polynomial-time algorithm
that finds a vector of taxes $\mathbf{\tau}=\left(\tau_{i,x}\right)_{i\in{\cal S},x\in\Psi_{i}}$
such that:
\begin{enumerate}
\item $-\epsilon<\tau_{i,x}<\epsilon$ (taxes are small)
\item \label{enu:monotonicity_of_taxes}$\tau_{i,x}>\tau_{i,x'}$ if $x\succ_{i}x'$
(taxes prefer more-preferred bundles)
\item \label{enu:bound_on_budgets}$1\leq\min_{i,x}\left\{ b_{i}+\tau_{i,x}\right\} \leq\max_{i,x}\left\{ b_{i}+\tau_{i,x}\right\} \leq1+\beta$
(inequality bound is preserved)
\item $b_{i}+\tau_{i,x}\neq b_{i'}+\tau_{i',x'}$ for $\left(i,x\right)\neq\left(i',x'\right)$
(no two perturbed prices are equal)
\item \label{enu:general-position}there is no price $\mathbf{\tilde{p}}\in\cal{{P}}$
at which any subset of linearly dependent budget-constraint hyperplanes
intersect%
\footnote{The original existence proof of \citet{ACEEI_Bud11} requires only
that no more than $M$ hyperplanes intersect at any point; this causes
problems in the conditional expectation argument \citep[Step 5]{ACEEI_Bud11}.%
}%
\ignore{
\footnote{In the case where each student may be interested in an exponential
number of bundles, we relax this requirement to say that no more than
$M$ perturbed budget-constraint hyperplanes \emph{from different
students} intersect. A careful analysis of Step 8 shows that this
suffices. %
}
}
{} 
\end{enumerate}
\end{lemma}
\begin{proof}
\ignore{
The trickiest property to satisfy is of course \ref{enu:general-position},
which requires that every $\left(M+1\right)$-tuple of hyperplanes
do not intersect at a single price vector $\mathbf{\tilde{p}}$. 
}

Assume wlog that $\mathbf{b}$ is rounded to the nearest $\bar{\beta}M^{-M}$:
otherwise we can include this rounding in the taxes. 

We proceed by induction on the pairs $\left(i,x\right)$ of students
and bundles: at each step let $\tau_{i,x}$ be much smaller than all
the taxes introduced so far%
\footnote{Assume wlog that for each $i$ we consider the $\left(i,x\right)$'s
in order reversed with respect to $\succsim_{i}$, so that property
\ref{enu:monotonicity_of_taxes} is guaranteed.%
}. 

More precisely, if $\left(i,x\right)$ is the $\nu^{\mbox{th}}$ pair
to be considered, then we set
\[
\tau_{i,x} \in \pm \bar{\beta}M^{-2\nu M},
\]

where the sign is chosen such that condition \ref{enu:bound_on_budgets} in the statement of the lemma is preserved.

Now, assume by contradiction that there exists a $k$-tuple $H\left(i_{1},x_{1},\tau_{i_{1},x_{1}}\right),\dots,H\left(i_{k},x_{k},\tau_{i_{k},x_{k}}\right)$
of hyperplanes that intersect at price vector $\mathbf{\tilde{p}}$,
and such that the $x_{i}$'s are linearly dependent. (Note that the
latter holds, in particular, for every $\left(M+1\right)$-tuple.) 

Assume further, wlog, that this is the first such $k$-tuple, with
respect to the order of the induction. In particular, this means that
$\left\{ x_{1},\dots,x_{_{k-1}}\right\} $ are linearly independent.
Now consider the system 
\[
\left(\begin{array}{ccc}
x_{1}^{T} & \dots & x_{k-1}^{T}\end{array}\right)\left(\alpha\right)=\left(x_{k}\right)
\]
Notice that it has rank $k-1$. We can now take $k-1$ linearly independent
rows $j_{1},\dots j_{k-1}$ such that the following system has the
same unique solution $\alpha$:
\[
\left(\begin{array}{ccc}
x_{1,j_{1}} & \dots & x_{k-1,j_{1}}\\
\vdots & \ddots\\
x_{1,j_{k-1}} &  & x_{k-1,j_{k-1}}
\end{array}\right)\left(\alpha\right)=\left(\begin{array}{c}
x_{k,j_{1}}\\
\vdots\\
x_{k,j_{k-1}}
\end{array}\right)
\]
Denote
\[
X=\left(\begin{array}{ccc}
x_{1,j_{1}} & \dots & x_{k-1,j_{1}}\\
\vdots & \ddots\\
x_{1,j_{k-1}} &  & x_{k-1,j_{k-1}}
\end{array}\right)
\]
Since $X$ is a square matrix of full rank it is invertible, so we have that
\[
\alpha=X^{-1}\left(\begin{array}{c}
x_{k,j_{1}}\\
\vdots\\
x_{k,j_{k-1}}
\end{array}\right)
\]
Now, recall that 
\[
X^{-1}=\frac{1}{\det X}\left(\begin{array}{ccc}
X_{1,1} & \dots & X_{k-1,1}\\
\vdots & \ddots\\
X_{1,k-1} &  & X_{k-1,k-1}
\end{array}\right)
\]
 where $X_{i,j}$ is the $\left(i,j\right)$-cofactor of $X$. Finally,
since $X$ is a Boolean matrix, its determinant and all of its cofactors
are integers of magnitude less than $\left(k-1\right)^{k-1}$. The
entries of $\alpha$ are therefore rational fractions with numerators
and denominators of magnitude less than $\left(k-1\right)^{k-1}$.

Now, by our assumption by contradiction, $k$ hyperplanes intersect
at $\mathbf{\tilde{p}}$:

\[
\left(\begin{array}{c}
x_{1}\\
\vdots\\
x_{k}
\end{array}\right)\left(\mathbf{\tilde{p}}\right)=\left(\begin{array}{c}
b_{i_{1}}+\tau_{i_{1,}x_{1}}\\
\vdots\\
b_{i_{k}}+\tau_{i_{k,}x_{k}}
\end{array}\right)
\]
Therefore,
\begin{gather}
b_{i_{k}}+\tau_{i_{k,}x_{k}}=x_{k}\cdot\mathbf{\tilde{p}}=\sum_{l=1}^{k-1}\alpha_{l}\left(x_{l}\cdot\mathbf{\tilde{p}}\right)=\sum_{l=1}^{k-1}\alpha_{l}\left(b_{i_{l}}+\tau_{i_{l,}x_{l}}\right)\label{eq:contradicted}
\end{gather}
However, if $\left(i_{k},x_{k}\right)$ is the $\nu^{\mbox{th}}$
pair added by the induction, then the following is an integer: 
\[
\sum_{l=1}^{k-1}\left(M^{M}\alpha_{l}\right)\cdot\frac{M^{2\left(\nu-1\right)M}}{\bar{\beta}}\left(b_{i_{l}}+\tau_{i_{l,}x_{l}}\right)
\]
but $\frac{M^{\left(2\nu-1\right)M}}{\bar{\beta}}\cdot\left(b_{i_{k}}+\tau_{i_{k,}x_{k}}\right)$
is not an integer, a contradiction to Equation (\ref{eq:contradicted}).
\end{proof}

\subsection{Finding a fixed point (steps 2-4)}

This subsection describes the price adjustment correspondence of \cite{ACEEI_Bud11},
and is brought here mostly for completeness.

We first define the price adjustment function:
\[
f\left(\mathbf{\tilde{p}}\right)=t\left(\mathbf{\tilde{p}}\right)+\frac{1}{2N}\mathbf{z}\left(t\left(\mathbf{\tilde{p}}\right);\mathbf{b},\mathbf{\tau}\right)
\]
Observe that if $\mathbf{\tilde{p}^{*}}$ is a fixed point $\mathbf{\tilde{p}^{*}}=f\left(\mathbf{\tilde{p}^{*}}\right)$
of $f$, then its truncation $t\left(\mathbf{\tilde{p}^{*}}\right)=\mathbf{p^{*}}$
defines an exact competitive equilibrium%
\footnote{See Appendix A, Step 2 of \cite{ACEEI_Bud11} for more details.%
}. Yet, we know that the economy may not have an exact equilibrium
- and indeed $f$ is discontinuous at the budget constraint hyperplanes,
and so it is not guaranteed to have a fixed point. 

Instead, we define an upper hemicontinuous, set-valued ``convexification''
of $f$:
\[
F\left(\mathbf{p}\right)=co\left\{ \mathbf{y}\colon\exists\mbox{ a sequence \ensuremath{\mathbf{p^{w}}\rightarrow\mathbf{p}}, \ensuremath{\mathbf{p}\neq\mathbf{p^{w}}\in\mathcal{P}\,}such that \ensuremath{f\left(\mathbf{p^{w}}\right)\rightarrow\mathbf{y}}}\right\} 
\]
The correspondence $F$ is upper hemicontinuous, non-empty, and convex;
therefore, by Kakutani's fixed point theorem it has a fixed point. 

Finally, by \cite{PPAD_Pap94} finding this fixed point of $F$ is
in PPAD. 
\begin{rem}
\textbf{Computing the Correspondence}:\textbf{ \label{Computing-the-Correspondence:}}

We round all price vectors to a $\left(\bar{\beta}M^{\frac{1}{2}-2\left(\nu_{\max}+1\right)M}\right)$-grid
(this precision suffices to implement the algorithm in lemma \ref{lem:general_position}). 

From the proof of \cite{PPAD_Pap94} it follows that it suffices to
compute just a single point in $F\left(\mathbf{p}\right)$ for every
$\mathbf{p}$ (this is important because the number points in $F\left(\mathbf{p}\right)$
on the grid may be exponential). At any point on the grid, the price
of any bundle is an integer multiple of $\left(\bar{\beta}M^{\frac{1}{2}-2\left(\nu_{\max}+1\right)M}\right)$.
In particular, any budget-constraint hyperplane which does not contain
$\mathbf{p}$, must be at distance at least $\left(\bar{\beta}M^{\frac{1}{2}-2\left(\nu_{\max}+1\right)M}\right)$.
Therefore, we can take any point $\mathbf{p'}$ at distance $\frac{1}{2}\left(\bar{\beta}M^{\frac{1}{2}-2\left(\nu_{\max}+1\right)M}\right)$
from $\mathbf{p}$, and which does not lie on any of the hyperplanes
that contain $\mathbf{p}$. Because no budget-constraint hyperplanes
lie between $\mathbf{p'}$ and $\mathbf{p}$, it follows that $f\left(\mathbf{p'}\right)\in F\left(\mathbf{p}\right)$.
\end{rem}

\subsection{From a fixed point to approximate CEEI (steps 5-9)}
\begin{lemma}
Given a fixed point $\mathbf{p^{*}}$ of $F$, we can find in polynomial
time a vector of prices $\mathbf{p^{\phi'}}$ such that $\left\Vert \mathbf{z}\left(\mathbf{p^{\phi'}},\mathbf{b},\mathbf{\tau}\right)\right\Vert _{2}\leq\frac{\sqrt{\sigma M}}{2}$\end{lemma}
\begin{proof}
We use the method of conditional expectation to derandomize Step 8
of \cite{ACEEI_Bud11}.

Recall that by remark \ref{Computing-the-Correspondence:}, there
exists a neighborhood around $\mathbf{p^{*}}$ which does not intersect
any budget-constraint hyperplanes (beyond those that contain $\mathbf{p^{*}}$).
Let $1,\dots,L'$ be the indices of students whose budget-constraint
hyperplanes intersect at $\mathbf{p^{*}}$. For student $i\in\left[L'\right]$,
let $w_{i}$ be the number of corresponding hyperplanes $H\left(i,x_{i}^{1},\tau_{i,x_{i}^{1}}\right),\dots H\left(i,x_{i}^{w_{i}},\tau_{i,x_{i}^{w_{i}}}\right)$
intersecting at $\mathbf{p^{*}}$, and assume wlog that the superindices
of $x_{i}^{1},\dots x_{i}^{w_{i}}$ are ordered according to $\succsim_{i}$. 

Let $d_{i}^{0}$ be agent $i$'s demand when prices are slightly perturbed
from $\mathbf{p^{*}}$ such that all $x_{i}^{j}$'s are affordable.
Such a perturbation exists and is easily computable because the hyperplanes
are linearly independent%
\footnote{This appears to be a slight inaccuracy in the proof in \cite{ACEEI_Bud11}%
}. Similarly, let $d_{i}^{1}$ denote agent $i$'s demand when $x_{i}^{2},\dots x_{i}^{w_{i}}$
are affordable, but $x_{i}^{1}$ is not, and so on. Finally, let $z_{S\setminus\left[L'\right]}\left(\mathbf{p^{*}},\mathbf{b},\mathbf{\tau}\right)=d_{S\setminus\left[L'\right]}\left(\mathbf{p^{*}},\mathbf{b},\mathbf{\tau}\right)-\mathbf{q}$
be the market clearing error when considering the rest of the students.
(The demands of $S\setminus\left[L'\right]$ is constant in the small
neighborhood $\mathbf{p^{*}}$ which does not intersect any additional
hyperplanes.)

By Lemma 3 of \cite{ACEEI_Bud11}, there exist distributions $a_{i}^{f}$
over $d_{i}^{f}$: 
\begin{eqnarray*}
a_{i}^{f}\in\left[0,1\right] & \forall & i\in\left[L'\right],\forall f\in\left\{ 0\right\} \cup\left[w_{i}\right]\\
\sum_{f=0}^{w_{i}}a_{i}^{f}=1 & \forall & i\in\left[L'\right]
\end{eqnarray*}
such that the clearing error of the \emph{expected demand} is $0$:
\[
z_{S\setminus\left[L'\right]}\left(\mathbf{p^{*}},\mathbf{b},\mathbf{\tau}\right)+\sum_{i=1}^{L'}\sum_{f=0}^{w_{i}}a_{i}^{f}d_{i}^{f}=0
\]

We first find such $a_{i}^{f}$ in polynomial time using linear programming.

The existence proof then considers, for each $i$, a random vector
$\Theta_{i}=\left(\Theta_{i}^{1},\dots,\Theta_{i}^{w_{i}}\right)$:
the vectors are independent and in any realization $\theta_{i}$ satisfy
$\sum_{f=0}^{w_{i}}\theta_{i}^{f}=1$, while the variables each have
support $\mbox{supp}\left(\Theta_{i}^{f}\right)=\left\{ 0,1\right\} $,
and expectation $\mbox{E}\left[\Theta_{i}^{f}\right]=a_{i}^{f}$. 

By Lemma 4 of \cite{ACEEI_Bud11}, the \emph{expected clearing error}
is bounded by:
\[
\mbox{E}_{\Theta_{!}\dots\Theta_{L'}}\left\Vert \sum_{i=1}^{L'}\sum_{f=0}^{w_{i}}\left(a_{i}^{f}-\theta_{i}^{f}\right)d_{i}^{f}\right\Vert _{2}^{2}=\sum_{i=1}^{L'}\mbox{E}_{\Theta_{i}}\left\Vert \sum_{f=0}^{w_{i}}\left(a_{i}^{f}-\theta_{i}^{f}\right)d_{i}^{f}\right\Vert _{2}^{2}\leq\frac{\sigma M}{4}
\]
We now proceed by induction on the students. For each $i$, if the
conditional expectation on $\left(\hat{\theta}_{j}\right)_{j<i}$
satisfies 
\[
\mbox{E}_{\Theta_{i}\dots\Theta_{L'}}\left[\left\Vert \sum_{i=1}^{L'}\sum_{f=0}^{w_{i}}\left(a_{i}^{f}-\theta_{i}^{f}\right)d_{i}^{f}\right\Vert _{2}^{2}\mid\hat{\theta}_{1},\dots,\hat{\theta}_{i-1}\right]\leq\frac{\sigma M}{4}
\]

then at least one $\hat{\theta}_{i}$ must also satisfy the above
bound. We can find such $\hat{\theta}_{i}$ in polynomial time by
computing the conditional expectation for every feasible $\hat{\theta}_{i}^{'}$:
\begin{align*}
\mbox{E}_{\Theta_{i+1}\dots\Theta_{L'}}\left[\left\Vert \sum_{j=1}^{L'}\sum_{f=0}^{w_{j}}\left(a_{j}^{f}-\theta_{j}^{f}\right)d_{j}^{f}\right\Vert _{2}^{2}\mid\hat{\theta}_{1},\dots,\hat{\theta}_{i}\right]  = &
 \sum_{j=1}^{i}\left\Vert \sum_{f=0}^{w_{j}}\left(a_{j}^{f}-\hat{\theta}_{j}^{f}\right)d_{j}^{f}\right\Vert _{2}^{2}\\
   & +\sum_{j=i+1}^{L'}\mbox{E}_{\Theta_{j}}\left\Vert \sum_{f=0}^{w_{j}}\left(a_{j}^{f}-\theta_{j}^{f}\right)d_{j}^{f}\right\Vert _{2}^{2}\\
   & +\sum_{j\neq h\leq i}\sum_{f=0}^{w_{j}}\sum_{g=0}^{w_{h}}\left(a_{j}^{f}-\hat{\theta}_{j}^{f}\right)\left(a_{h}^{g}-\hat{\theta}_{h}^{g}\right)
\end{align*}
\;

\end{proof}
The chosen $\left(\hat{\theta}_{i}\right)_{i=1}^{L'}$ define an allocation
$\mathbf{x^{*}}$ with bounded clearing error. We now follow step
9 of \cite{ACEEI_Bud11} in order to define budgets $\mathbf{b^{*}}$
such that $\mathbf{x^{*}}$ is the preferred consumption by all the
students at price $\mathbf{p^{*}}$.

We define, for every $i$, $b_{i}^{*}=b_{i}+\tau_{i,x_{i}^{*}}$.
For $i>L'$ we have $x_{i}^{*}=d_{i}\left(\mathbf{p^{*}},b_{i},\tau_{i}\right)$.
By requirement \ref{enu:monotonicity_of_taxes} of lemma \ref{lem:general_position},
every bundle that student $i$ prefers over $x_{i}^{*}$ had a greater
tax and was still unaffordable at $\mathbf{p^{*}}$; it now costs
more than $b_{i}+\tau_{i,x_{i}^{*}}$.

For $i\leq L'$ notice that every bundle $x_{i}^{\perp}$ that $i$
prefers over $x_{i}^{*}$ and was exactly affordable at $\mathbf{p^{*}}$
with taxes $\mathbf{\tau}$ and budget $\mathbf{b}$, $x^{\perp}$
must cost strictly more than $i$'s new budget $b_{i}^{*}$. Therefore,
$\left(\mathbf{x^{*}},\mathbf{b^{*}},\mathbf{p^{*}}\right)$ is a
$\left(\frac{\sqrt{\sigma M}}{2},\beta\right)$-CEEI 

\qed

\section{Additional gadgets for Theorem \ref{thm:ppad-hard}}

In this section we prove Lemma \ref{lem:additional-gadgets}.\\

\begin{lemma_again}{\ref{lem:additional-gadgets}}
Let $n_{x}\geq2^{8}\cdot\alpha$ and suppose that the economy has
courses $c_{x}$ and $c_{y}$. Then for any of the functions $f$
listed below, we can add: a course $c_{z}$, and at most $n_{x}$
students interested in each of $c_{x}$ and $c_{y}$, such that in
any $\left(\alpha,\beta\right)$-CEEI $p_{z}^{*}\in\left[f\left(p_{x}^{*},p_{y}^{*}\right)-2\beta,f\left(p_{x}^{*},p_{y}^{*}\right)+2\beta\right]$ 
\begin{enumerate}
\item HALF: $f_{G_{/2}}\left(x\right)=x/2$
\item VALUE: $f_{G_{\frac{1}{2}}} \equiv\frac{1}{2}$
\item SUM: $f_{G_{+}}\left(x,y\right)=\min\left(x+y,1\right)$
\item DIFF: $f_{G_{-}}\left(x,y\right)=\max\left(x-y,0\right)$
\item LESS: $f_{G_{<}}\left(x,y\right)=\begin{cases}
1 & x>y+\beta\\
0 & y>x+\beta
\end{cases}$
\item AND: $f_{G_{\wedge}}\left(x,y\right)=\begin{cases}
1 & \left(x>\frac{1}{2}+\beta\right)\wedge\left(y>\frac{1}{2}+\beta\right)\\
0 & \left(x<\frac{1}{2}-\beta\right)\vee\left(y<\frac{1}{2}-\beta\right)
\end{cases}$
\item OR: $f_{G_{\vee}}\left(x,y\right)=\begin{cases}
1 & \left(x>\frac{1}{2}+\beta\right)\vee\left(y>\frac{1}{2}+\beta\right)\\
0 & \left(x<\frac{1}{2}-\beta\right)\wedge\left(y<\frac{1}{2}-\beta\right)
\end{cases}$
\end{enumerate}
In particular, $p_{z}^{*}\in\left[f\left(p_{x}^{*},p_{y}^{*}\right)-2\beta,f\left(p_{x}^{*},p_{y}^{*}\right)+2\beta\right]$
in every $\left(\alpha,\beta\right)$-CEEI even if up to $n_{z}\leq n_{x}/2^{8}$
additional students (beyond the ones specified in the proofs below)
are interested in course $c_{z}$.
\end{lemma_again}

Notice, that like in similar gadget reductions from $\PPAD$-complete
problems, LESS, AND, and OR are brittle comparators (see discussion
in \citet{NASH-is-PPAD-hard_DGP09} for more details).
\begin{proof}
~
\begin{enumerate}
\item HALF:

Let $c_{z}$ have capacity $q_{z}=n_{x}/8$, let $n_{z}=q_{z}/2$,
and consider three auxiliary courses $c_{1}$, $c_{2}$, and $c_{\overline{x}}$
of capacities $q_{1}=q_{2}=q_{z}$ and $q_{\overline{x}}=n_{x}/2$.
Using lemma \ref{lem:not_gadget} add $n_{x}$ students that will
guarantee $p_{\overline{x}}\in\left[1-p_{x}^{*},1-p_{x}^{*}+\beta\right]$.
Additionally, consider $n_{\overline{x}}=n_{x}/4$ students with preference
list: $\left(\left\{ c_{z},c_{1},c_{\overline{x}}\right\} ,\left\{ c_{z},c_{2},c_{\overline{x}}\right\} ,\left\{ c_{1},c_{2},c_{\overline{x}}\right\} \right)$
(in this order), then:
\begin{itemize}
\item If the total price $p_{i}^{*}+p_{j}^{*}$ of any pair $i,j\in\left\{ 1,2,z\right\} $
is less than $p_{x}^{*}-\beta$, then all $n_{\overline{x}}$ students
will be able to afford some subset in their preference list, leaving
a total overbooking of at least $z_{z}^{*}+z_{1}^{*}+z_{2}^{*}\geq2n_{\overline{x}}-3q_{z}=n_{x}/8$,
which violates the $\left(\alpha,\beta\right)$-CEEI conditions
\item If the total price of any of the pairs above (wlog, $p_{1}^{*}+p_{2}^{*}$)
is greater than $p_{x}^{*}+\beta$, then none of the $n_{\overline{x}}$
students will be able to afford the subset $\left\{ c_{1},c_{2},c_{\overline{x}}\right\} $.
Therefore the number of students taking $c_{z}$ will be at least
the sum of students taking $c_{1}$ or $c_{2}$. Therefore, even after
taking into account $n_{z}$ additional students, we have that $z_{z}^{*}+z_{1}^{*}+z_{2}^{*}\geq q_{z}-n_{z}=n_{x}/16$. 
\end{itemize}

\item VALUE:

Similarly to the HALF gadget, consider two auxiliary courses $c_{1}$
and $c_{2}$, and let $n_{x}$ students have preferences: $\left(\left\{ c_{z},c_{1}\right\} ,\left\{ c_{z},c_{2}\right\} ,\left\{ c_{1},c_{2}\right\} \right)$.
Then, following the argument for the HALF gadget, it is easy to see
that $p_{z}^{*}\in\left[\frac{1}{2},\frac{1}{2}+\beta\right]$ in
any $\left(\alpha,\beta\right)$-CEEI, with $n_{z}=n_{x}/8$.

\item DIFF:

Let $c_{\overline{x}}$ be a course with price $p_{\overline{x}}^{*}\in\left[1-p_{x}^{*},1-p_{x}^{*}+\beta\right]$,
$q_{\overline{x}}=n_{x}/2$, and consider $n_{\overline{x}}=n_{x}/4$
students willing to take $\left\{ c_{\overline{x}},c_{y},c_{z}\right\} $.
Then it is easy to see that 
\begin{align*}
p_{z}^{*} & \in 
	\left[ 1-p_{\overline{x}}^{*}-p_{y}^{*},
	1-p_{\overline{x}}^{*} - p_{y}^{*}+\beta\right]\\
& \subseteq\left[p_{x}^{*}-p_{y}^{*}-\beta,p_{x}^{*}-p_{y}^{*}+\beta\right]
\end{align*}
with $n_{z}=n_{x}/16$

\item SUM:

Concatenating NOT and DIFF gadgets, we have: 
\begin{flalign*}
p_{\overline{x}}^{*} & \in 
	\left[1-p_{x}^{*},1-p_{x}^{*}+\beta\right]\\
p_{\overline{z}}^{*} & \in 
	\left[p_{\overline{x}}^{*}-p_{y}^{*}-\beta,
		p_{\overline{x}}^{*}-p_{y}^{*}+\beta\right]\\
p_{z}^{*} & \in  
	\left[1-\left(p_{\overline{x}}^{*}-p_{y}^{*}+\beta\right),
		1-\left(p_{\overline{x}}^{*}-p_{y}^{*}-\beta\right)+\beta\right]\\
& \subseteq  \left[p_{x}^{*}+p_{y}^{*}-2\beta,
				 p_{x}^{*}+p_{y}^{*}+2\beta\right]
\end{flalign*}
for $n_{z}=n_{x}/2^{8}$

\item LESS:

Let $c_{\overline{x}}$ be a course with price $p_{\overline{x}}^{*}\in\left[1-p_{x}^{*},1-p_{x}^{*}+\beta\right]$,
$q_{\overline{x}}=n_{x}/2$; let $q_{z}=n_{x}/8$ and $n_{z}=n_{x}/16$.
Consider $n_{x}/4$ students wishing to take $\left(\left\{ c_{\overline{x}},c_{y}\right\} \left\{ c_{z}\right\} \right)$,
in this order:
\begin{itemize}
\item If $p_{y}^{*}>p_{x}^{*}+\beta$, then $p_{\overline{x}}^{*}+p_{y}^{*}>1+\beta$,
and therefore none of the $n_{x}/4$ students will be able to afford
the first pair; they will all try to sign up to $c_{z}$ which will
be overbooked unless $p_{z}^{*}>1$
\item If $p_{x}^{*}>p_{y}^{*}+\beta$, then all $n_{x}/4$ students will
sign up for the first pair, forcing $p_{z}^{*}=0$ in any $\left(\alpha,\beta\right)$-CEEI.
\end{itemize}
\item AND:

Let $c_{\frac{1}{2}}$ be a course with price $p_{\frac{1}{2}}^{*}\in\left[\frac{1}{2},\frac{1}{2}+\beta\right]$
and $n_{\frac{1}{2}}=n_{x}/8$, as guaranteed by gadget VALUE; let
$q_{z}=n_{x}/32$ and $n_{z}=n_{x}/64$. Consider $n_{x}/16$ students
wishing to take $\left(\left\{ c_{x},c_{\frac{1}{2}}\right\} ,\left\{ c_{y},c_{\frac{1}{2}}\right\} ,\left\{ c_{z}\right\} \right)$,
in this order. 
\begin{itemize}
\item If $\left(p_{x}^{*}>\frac{1}{2}+\beta\right)\wedge\left(p_{y}^{*}>\frac{1}{2}+\beta\right)$,
then the $n_{x}/16$ students can afford neither pair. They will all
try to sign up for $c_{z}$, forcing $p_{z}^{*}>1$, in any $\left(\alpha,\beta\right)$-CEEI.
\item If $\left(x<\frac{1}{2}-\beta\right)\vee\left(y<\frac{1}{2}-\beta\right)$,
then the $n_{x}/16$ students can afford at least one of the pairs
and will register for those courses. Thus $p_{z}^{*}=0$.
\end{itemize}
\item OR:

Similar to the AND gadget; students will want $\left(\left\{ c_{x},c_{y},c_{\frac{1}{2}}\right\} ,\left\{ c_{z}\right\} \right)$,
in this order.

\end{enumerate}
\end{proof}

\end{document}